\documentclass[submission,copyright,creativecommons]{eptcs}
\providecommand{\event}{NCMA 2022}

\usepackage[english]{babel}

\usepackage{hyperref}

\usepackage[T1]{fontenc}

\usepackage{float}

\usepackage{bbm}
\usepackage{amssymb, amsmath, amsthm}
\usepackage{scalerel}

\DeclareMathOperator*{\bigboxplus}{\scalerel*{\boxplus}{\sum}}

\usepackage{mathtools}

\newcommand{\dollarfmap}{%
  \mathbin{{<}\mspace{-4mu}{\$}\mspace{-4mu}{>}}%
}
\newcommand{\bind}{%
  \mathbin{{>}\mspace{-5mu}{>}\mspace{-4mu}{=}}%
}
\newcommand{\kleisli}{%
  \mathbin{{>}\mspace{-4mu}{=}\mspace{-3mu}{>}}%
}

\usepackage{tikz}
  \usetikzlibrary{automata,automata, positioning, arrows.meta}
  \usetikzlibrary{arrows}

  \usetikzlibrary{shapes,snakes}

    \usetikzlibrary{decorations.pathmorphing}
  \usetikzlibrary{fit}
  \usetikzlibrary{trees}
   \usepackage{soul}
  \usetikzlibrary{calc}
  \tikzstyle{every picture}=[
  >=stealth', shorten >=1pt, node distance=1.44cm,auto,bend angle=45,initial text=,
  every state/.style={inner sep=0.75mm, minimum size=1mm},font=\scriptsize,
  triangle/.style = { regular polygon, regular polygon sides=3, inner sep=0}
]

\newtheorem{example}{Example}
\newtheorem{definition}{Definition}
\newtheorem{proposition}{Proposition}
\newtheorem{theorem}{Theorem}

\begin{document}

\setlength{\abovedisplayskip}{0pt}
\setlength{\belowdisplayskip}{0pt}
\setlength{\abovedisplayshortskip}{0pt}
\setlength{\belowdisplayshortskip}{0pt}

\title{Monadic Expressions and their Derivatives}

\author{
  Samira Attou
  \institute{LITIS,\\
    Université de Rouen Normandie,\\
    Avenue de l'Université,\\
    76801 Saint-Étienne-du-Rouvray, France\\
    \email{samira.attou@univ-rouen.fr}}\\
  \and
  Ludovic Mignot
  \institute{GR\textsuperscript{2}IF,\\
    Université de Rouen Normandie,\\
    Avenue de l'Université,\\
    76801 Saint-Étienne-du-Rouvray, France \\
    \email{ludovic.mignot@univ-rouen.fr}}\\
  \and
  Clément Miklarz
  \institute{GR\textsuperscript{2}IF,\\
    Université de Rouen Normandie,\\
    Avenue de l'Université,\\
    76801 Saint-Étienne-du-Rouvray, France \\
    \email{clement.miklarz1@univ-rouen.fr}}\\
  \and
  Florent Nicart
  \institute{GR\textsuperscript{2}IF,\\
    Université de Rouen Normandie,\\
    Avenue de l'Université,\\
    76801 Saint-Étienne-du-Rouvray, France \\
    \email{florent.nicart@univ-rouen.fr}}}

\def\titlerunning{Monadic Expressions and their Derivatives}

\def\authorrunning{S.Attou et al.}

\maketitle

\begin{abstract}
  We propose another interpretation of well-known derivatives computations from regular expressions, due to Brzozowski, Antimirov or Lombardy and Sakarovitch,
  in order to abstract the underlying data structures (\emph{e.g.} sets or linear combinations) using the notion of monad.
  As an example of this generalization advantage, we introduce a new derivation technique based on the graded module monad.

  We also extend operators defining expressions to any \(n\)-ary functions over value sets, such as
  classical operations (like negation or intersection for Boolean weights) or more exotic ones (like algebraic mean for rational weights).

  Moreover, we present how to compute a (non-necessarily finite) automaton from such an extended expression, using the Colcombet and Petrisan categorical definition of automata.
  These category theory concepts allow us to perform this construction in a unified way, whatever the underlying monad.

  Finally, to illustrate our work, we present a Haskell implementation of these notions using advanced techniques of functional programming,
  and we provide a web interface to manipulate concrete examples.
\end{abstract}


\section{Introduction}

Regular expressions are a classical way to represent associations between words and value sets.
As an example, classical regular expressions denote sets of words and regular expressions with multiplicities denote formal series.
From a regular expression, solving the membership test (determining whether a word belongs to the denoted language) or the weighting test (determining the weight of a word in the denoted formal series) can be solved, following Kleene theorems~\cite{Kle56,Sch61} by computing a finite automaton, such as the position automaton~\cite{Glu61,BS86,CF11,CLOZ04}.

Another family of methods to solve these tests is the family of derivative computations, that does not require the construction of a whole automaton.
The common point of these techniques is to transform the test for an arbitrary word into the test for the empty word, which can be easily solved in a purely syntactical way (\emph{i.e.} by induction over the structure of expressions).
Brzozowski~\cite{Brzo64} shows how to compute, from a regular expression \(E\) and a word \(w\), a regular expression \(d_w(E)\) denoting the set of words \(w'\) such that \(ww'\) belongs to the language denoted by \(E\).
Solving the membership test hence becomes the membership test for the empty word in the expression \(d_w(E)\).
Antimirov~\cite{Ant96} modifies this method in order to produce sets of expressions instead of expressions, \emph{i.e.} defines the partial derivatives \(\partial_w(E)\) as a set of expressions the sum of which denotes the same language as \(d_w(E)\).
If the number of derivatives is exponential w.r.t.\ the length \(|E|\) of \(E\) in the worst case\footnote{as far as rules of associativity, commutativity and idempotence of the sum are considered, possibly infinite otherwise.}, the partial derivatives produce at most a linear number of expressions w.r.t. \(|E|\).
Finally, Lombardy and Sakarovitch~\cite{LS05} extends these methods to expressions with multiplicities.

It is well-known that these methods are based on a common operation, the quotient of languages.
Furthermore, Antimirov's method can be interpreted as the derivation of regular expression with multiplicities in the Boolean semiring.
However, the Brzozowski computation does not produce the same expressions (\emph{i.e.} equality over the syntax trees) as the Antimirov one.

\textbf{Main contributions:}
In this paper, we present a unification of these computations by applying notions of category theory to the category of sets,
and show how to compute categorical automata as defined in~\cite{CP17}, by reinterpreting the work started in~\cite{LM20}.
We make use of classical monads to model well-known derivatives computations.
Furthermore, we deal with \emph{extended} expressions in a general way: in this paper, expressions can support extended operators like complement, intersection, but also any \(n\)-ary function (algebraic mean, extrema multiplications, \emph{etc}.).
The main difference with~\cite{LM20} is that we formally state the languages and series that the expressions denote in an inherent way w.r.t.\ the underlying monads.

More precisely, this paper presents:
\begin{itemize}
    \item an extension of expressions to any \(n\)-ary function over the value set,
    \item a monadic generalization of expressions,
    \item a solution for the membership/weight test for these expressions,
    \item a computation of categorical derivative automata,
    \item a new monad that fits with the extension to \(n\)-ary functions,
    \item an illustration implemented in Haskell using advanced functional programming.
\end{itemize}

\textbf{Motivation:}
The unification of derivation techniques is a goal by itself.
Moreover, the formal tools used to achieve this unification are also useful:
Monads offer both theoretical and practical advantages.
Indeed, from a theoretical point of view, these structures allow the abstraction of properties
and focus on the principal mechanisms that allow solving the membership and weight problems.
Besides, the introduction of exotic monads can also facilitate the study of finiteness of derivated terms.
From a practical point of view, monads are easy to implement (even in some other languages than Haskell)
and allow us to produce compact and safe code. Finally, we can easily combine different algebraic structures or add some technical functionalities
(capture groups, logging, nondeterminism, \emph{etc.}) thanks to notions like monad transformers~\cite{MPJ95}.

This paper is structured as follows.
In Section~\ref{sec prelim}, we gather some preliminary material, like algebraic structures or category theory notions.
We also introduce some functions well-known to the Haskell community that can allow us to reduce the size of our equations.
We then structurally define the expressions we deal with, the associated series and the weight test for the empty word in Section~\ref{sec expr def}.
In order to extend this test to any arbitrary word, we first state in Section~\ref{sec supp} some properties required by the monads we consider.
Once this so-called support is determined, we show in Section~\ref{sec deriv} how to compute the derivatives.
The computation of derivative automata is explained in Section~\ref{sec aut cons}.
A new monad and its associated derivatives computation is given in Section~\ref{sec new mon}.
Finally, our implementation is presented in Section~\ref{sec haskell}.

\section{Preliminaries}\label{sec prelim}

We denote by \(S \rightarrow S'\) the set of functions from a set \(S\) to a set \(S'\).
The notation \(\lambda x \rightarrow f(x) \) is an equivalent notation for a function \(f\).

A \emph{monoid} is a set \(S\) endowed with an associative operation and a unit element.
A \emph{semiring} is a structure \((S, \times, +, 1, 0)\) such that \((S, \times, 1)\) is a monoid, \((S, +, 0)\) is a commutative monoid, \(\times \) distributes over \(+\) and \(0\) is an annihilator for \(\times \).
A \emph{starred semiring} is a semiring with a unary function \({}^\star \) such that
\begin{equation*}
    k^\star = 1 + k\times k^\star = 1 + k^\star\times k.
\end{equation*}
%

%

A \(\mathbb{K}\)-\emph{series} over the free monoid \((\Sigma^*, \cdot, \varepsilon)\) associated with an alphabet \(\Sigma \), for a semiring \(\mathbb{K}\)  where \(\mathbb{K}=(K, \times, +, 1, 0)\), is a function from \(\Sigma^*\) to \(K\).
The set of \(\mathbb{K}\)-\emph{series} can be endowed with the structure of semiring as follows:
\begin{align*}
    1(w)                               & =
    \begin{cases}
        1 & \text{if } w = \varepsilon, \\
        0 & \text{otherwise},
    \end{cases} &
    0(w)                               & = 0,                                          \\
    (S_1 + S_2)(w)                     & = S_1(w) + S_2(w),                          &
    (S_1 \times S_2)(w)                & = \sum_{u\cdot v = w} S_1(u) \times S_2(v).
\end{align*}
Furthermore, if \(S_1(\varepsilon) = 0\) (\emph{i.e.} \(S_1\) is said to be \emph{proper}), the \emph{star} of \(S_1\) is the series defined by
\begin{align*}
    {(S_1)}^\star(\varepsilon) & = 1,                                                                                                   &
    {(S_1)}^\star(w)           & = \sum_{n \leq |w|, w = u_1 \cdots u_n, u_j \neq \varepsilon } S_1(u_1) \times \cdots \times S_1(u_n).
\end{align*}
Finally, for any function \(f\) in \(K^n \rightarrow K\), we set:
\begin{equation}
    (f(S_1, \ldots, S_n))(w) = f(S_1(w), \ldots, S_n(w)). \label{eq f series}
\end{equation}

A \emph{functor}\footnote{More precisely, a functor over a subcategory of the category of sets.} \(F\) associates
with each set \(S\) a set \(F(S)\)
and with each function \(f\) in \(S \rightarrow S'\) a function \(F(f)\) from \(F(S)\) to \(F(S')\)
such that
\begin{align*}
    F(\mathrm{id}) & = \mathrm{id}, & F(f \circ g) & = F(f) \circ F(g),
\end{align*}
where \(\mathrm{id}\) is the identity function and \(\circ \) the classical function composition.

\noindent A \emph{monad}\footnote{More precisely, a monad over a subcategory of the category of sets.} \(M\) is a functor endowed with two (families of) functions
\begin{itemize}
    \item \(\mathtt{pure}\), from a set \(S\) to \(M(S)\),
    \item \(\mathtt{bind}\), sending any function \(f\) in \(S \rightarrow M(S')\) to \(M(S) \rightarrow M(S')\),
\end{itemize}
such that the three following conditions are satisfied:
\begin{gather*}
    \begin{aligned}
        \mathtt{bind}(f)(\mathtt{pure}(s)) & = f (s),       &
        \mathtt{bind}(\mathtt{pure})       & = \mathrm{id},
    \end{aligned}\\
    \mathtt{bind}(g)(\mathtt{bind}(f)(m)) = \mathtt{bind}(\lambda x \rightarrow \mathtt{bind}(g) (f(x)))(m).
\end{gather*}

\begin{example}
    The \(\mathtt{Maybe}\) monad associates:
    \begin{itemize}
        \item any set \(S\) with the set \(\mathtt{Maybe}(S) = \{\mathtt{Just}(s) \mid s \in S\} \cup \{\mathtt{Nothing}\} \), where \(\mathtt{Just}\) and \(\mathtt{Nothing}\) are two syntactic tokens allowing us to extend a set with one value;
        \item any function \(f\) with the function \(\mathtt{Maybe}(f)\) defined by
              \begin{align*}
                  \mathtt{Maybe}(f)(\mathtt{Just}(s)) & = \mathtt{Just} (f(s)), &
                  \mathtt{Maybe}(f)(\mathtt{Nothing}) & = \mathtt{Nothing}
              \end{align*}
        \item is endowed with the functions \(\mathtt{pure}\) and \(\mathtt{bind}\) defined by:
              \begin{gather*}
                  \begin{aligned}
                      \mathtt{pure}(s) & = \mathtt{Just}(s),
                  \end{aligned}
                  \qquad\qquad
                  \begin{aligned}
                      \mathtt{bind}(f)(\mathtt{Just}(s)) & = f(s),             \\
                      \mathtt{bind}(f)(\mathtt{Nothing}) & = \mathtt{Nothing}.
                  \end{aligned}
              \end{gather*}
    \end{itemize}
\end{example}

\begin{example}
    The \(\mathtt{Set}\) monad associates:
    \begin{itemize}
        \item with any set \(S\) the set \(2^S\),
        \item with any function \(f\) the function \(\mathtt{Set}(f)\) defined by \( \mathtt{Set}(f)(R) = \bigcup_{r\in R} \{f(r)\}, \)
        \item is endowed with the functions \(\mathtt{pure}\) and \(\mathtt{bind}\) defined by:
              \begin{align*}
                  \mathtt{pure}(s)    & = \{s\},                &
                  \mathtt{bind}(f)(R) & = \bigcup_{r\in R}f(r).
              \end{align*}
    \end{itemize}
\end{example}

\begin{example}
    The \(\mathtt{LinComb}(\mathbb{K})\) monad, for \(\mathbb{K}=(K, \times, +, 1, 0)\), associates:
    \begin{itemize}
        \item with any set \(S\) the set of \(\mathbb{K}\)-linear combinations of elements of \(S\), where a linear combination is a finite (formal, commutative) sum of couples (denoted by \(\boxplus \)) in \(K\times S\) where \((k, s) \boxplus (k', s) = (k + k', s)\),
        \item with any function \(f\) the function \(\mathtt{LinComb}(\mathbb{K})(f)\) defined by
              \begin{equation*}
                  \mathtt{LinComb}(\mathbb{K})(f)(R) = \bigboxplus_{(k, r)\in R} (k, f(r)),
              \end{equation*}
        \item is endowed with the functions \(\mathtt{pure}\) and \(\mathtt{bind}\) defined by:
              \begin{align*}
                  \mathtt{pure}(s)    & = (1, s),                                   &
                  \mathtt{bind}(f)(R) & = \bigboxplus_{(k, r)\in R} k \otimes f(r),
              \end{align*}
              where \(\displaystyle k \otimes R = \bigboxplus_{(k', r) \in R}(k \times k', r)\).
    \end{itemize}
\end{example}
To compact equations, we use the following operators for any monad \(M\):
\begin{align*}
    f \dollarfmap s & = M(f)(s),             &
    m \bind f       & = \mathtt{bind}(f)(m).
\end{align*}
If \(\dollarfmap \) can be used to lift unary functions to the monadic level, \(\bind \) and \(\mathtt{pure}\) can be used to lift any \(n\)-ary function \(f\) in \(S_1\times \cdots \times S_n \rightarrow S\), defining a function \(\mathtt{lift}_n\) sending \(S_1\times \cdots \times S_n \rightarrow S\) to \(M(S_1)\times \cdots \times M(S_n) \rightarrow M(S)\) as follows:
\begin{align*}
    \mathtt{lift}_n(f)(m_1, \ldots, m_n) = &
    m_1 \bind (\lambda s_1 \rightarrow           \ldots                                                                           \\
                                           & \qquad m_n \bind (\lambda s_n \rightarrow \mathtt{pure}(f(s_1, \ldots, s_n)))\ldots)
\end{align*}

Let us consider the set \(\mathbbm{1}=\{\top \} \) with only one element.
The images of this set by some previously defined monads can be evaluated as value sets classically used to weight words in association with classical regular expressions.
As an example, \(\mathtt{Maybe}(\mathbbm{1})\) and \(\mathtt{Set}(\mathbbm{1})\) are isomorphic to the Boolean set, and any set \(\mathtt{LinComb}(\mathbb{K})(\mathbbm{1})\) can be converted into the underlying set of \(\mathbb{K}\).
This property allows us to extend in a coherent way classical expressions to monadic expressions, where the type of the weights is therefore given by the ambient monad.

\section{Monadic Expressions}\label{sec expr def}

As seen in the previous section, elements in \(M(\mathbbm{1})\) can be evaluated as classical value sets for some particular monads.
Hence, we use these elements not only for the weights associated with words by expressions, but also for the elements that act over the denoted series.

In the following, in addition to classical operators (\(+\), \(\cdot \) and \({}^*\)), we denote:
\begin{itemize}
    \item the action of an element over a series by \(\odot \),
    \item the application of a function by itself.
\end{itemize}
\begin{definition}
    Let \(M\) be a monad.
    An \(M\)-\emph{monadic expression} \(E\) over an alphabet \( \Sigma \) is inductively defined as follows:
    \begin{align*}
        E & = a,                & E & = \varepsilon,      & E & = \emptyset,                      \\
        E & = E_1 + E_2,        & E & = E_1 \cdot E_2,    & E & = E_1^*,                          \\
        E & = \alpha \odot E_1, & E & = E_1 \odot \alpha, & E & = f\left(E_1, \ldots, E_n\right),
    \end{align*}
    where \(a\) is a symbol in \(\Sigma \), \((E_1, \ldots, E_n)\) are \(n\) \(M\)-monadic expressions over \( \Sigma \), \( \alpha \) is an element of \(M(\mathbbm{1})\) and \(f\) is a function from \({(M(\mathbbm{1}))}^n\) to \(M(\mathbbm{1})\).
\end{definition}
We denote by \(\mathrm{Exp}(\Sigma)\) the set of monadic expressions over an alphabet \( \Sigma \).

\begin{example}\label{ex:def extdist}
    As an example of functions that can be used in our extension of classical operators,
    one can define the function \(\mathtt{ExtDist}(x_1, x_2, x_3) = \max(x_1, x_2, x_3) - \min(x_1, x_2, x_3)\) from \(\mathbb{N}^3\) to \(\mathbb{N}\).
\end{example}

Similarly to classical regular expressions, monadic expressions associate a weight with any word.
Such a relation can be denoted \emph{via} a formal series.
However, before defining this notion, in order to simplify our study, we choose to only consider proper expressions.
Let us first show how to characterize them by the computation of a nullability value.
\begin{definition}
    Let \(M\) be a monad such that the structure \((M(\mathbbm{1}), +, \times, {}^\star, 1, 0)\) is a starred semiring.
    The \emph{nullability value} of an \(M\)-monadic expression \(E\) over an alphabet \( \Sigma \) is the element \(\mathtt{Null}(E)\) of \( M(\mathbbm{1}) \) inductively defined as follows:
    \begin{gather*}
        \begin{aligned}
            \mathtt{Null}(\varepsilon)      & = 1,                                            &
            \mathtt{Null}(\emptyset)        & = 0,                                              \\
            \mathtt{Null}(a)                & = 0,                                            &
            \mathtt{Null}(E_1 + E_2)        & = \mathtt{Null}(E_1) + \mathtt{Null}(E_2),        \\
            \mathtt{Null}(E_1 \cdot E_2)    & = \mathtt{Null}(E_1) \times \mathtt{Null}(E_2), &
            \mathtt{Null}(E_1^*)            & = {\mathtt{Null}(E_1)}^\star,                     \\
            \mathtt{Null}(\alpha \odot E_1) & = \alpha \times \mathtt{Null}(E_1),             &
            \mathtt{Null}(E_1 \odot \alpha) & = \mathtt{Null}(E_1) \times \alpha,
        \end{aligned}\\
        \mathtt{Null}(f(E_1, \ldots, E_n)) = f(\mathtt{Null}(E_1), \ldots, \mathtt{Null}(E_n)),
    \end{gather*}
    where \(a\) is a symbol in \(\Sigma \), \((E_1, \ldots, E_n)\) are \(n \) \(M\)-monadic expressions over \( \Sigma \), \( \alpha \) is an element of \(M(\mathbbm{1})\) and \(f\) is a function from \({(M(\mathbbm{1}))}^n\) to \(M(\mathbbm{1})\).
\end{definition}
When the considered semiring is not a starred one, we restrict the nullability value computation to expressions where a starred subexpression admits a \emph{null nullability value}.
In order to compute it, let us consider the Maybe monad, allowing us to elegantly deal with such a partial function.
\begin{definition}\label{def partial nullability}
    Let \(M\) be a monad such that the structure \((M(\mathbbm{1}), +, \times, 1, 0)\) is a semiring.
    The \emph{partial nullability value} of an \(M\)-monadic expression \(E\) over an alphabet \( \Sigma \) is the element \(\mathtt{PartNull}(E)\) of \( \mathtt{Maybe}(M(\mathbbm{1})) \) defined as follows:
    \begin{gather*}
        \begin{aligned}
            \mathtt{PartNull}(\varepsilon) & = \mathtt{Just}(1), &
            \mathtt{PartNull}(\emptyset)   & = \mathtt{Just}(0), &
            \mathtt{PartNull}(a)           & = \mathtt{Just}(0),
        \end{aligned}\\
        \begin{aligned}
            \mathtt{PartNull}(E_1 + E_2)           & = \mathtt{lift}_2(+) (\mathtt{PartNull}(E_1), \mathtt{PartNull}(E_2)),        \\
            \mathtt{PartNull}(E_1 \cdot E_2)       & = \mathtt{lift}_2(\times )(\mathtt{PartNull}(E_1), \mathtt{PartNull}(E_2)),   \\
            \mathtt{PartNull}(E_1^*)               & =
            \begin{cases}
                \mathtt{Just}(1) & \text{if } \mathtt{PartNull}(E_1) = \mathtt{Just}(0), \\
                \mathtt{Nothing} & \text{otherwise,}
            \end{cases}                                               \\
            \mathtt{PartNull}(\alpha \odot E_1)    & = (\lambda E \rightarrow \alpha \times E) \dollarfmap \mathtt{PartNull}(E_1), \\
            \mathtt{PartNull}(E_1 \odot \alpha)    & = (\lambda E \rightarrow E \times \alpha) \dollarfmap \mathtt{PartNull}(E_1), \\
            \mathtt{PartNull}(f(E_1, \ldots, E_n)) & = \mathtt{lift}_n(f)(\mathtt{PartNull}(E_1), \ldots, \mathtt{PartNull}(E_n)),
        \end{aligned}
    \end{gather*}
    where \(a\) is a symbol in \(\Sigma \), \((E_1, \ldots, E_n)\) are \(n \) \(M\)-monadic expressions over \( \Sigma \), \( \alpha \) is an element of \(M(\mathbbm{1})\) and \(f\) is a function from \({(M(\mathbbm{1}))}^n\) to \(M(\mathbbm{1})\).
\end{definition}
An expression \(E\) is \emph{proper} if its partial nullability value is not \(\mathtt{Nothing}\), therefore if it is a value \(\mathtt{Just}(v)\);
in this case, \(v\) is its nullability value, denoted by \(\mathtt{Null}(E)\) (by abuse).
\begin{definition}\label{def series}
    Let \(M\) be a monad such that the structure \((M(\mathbbm{1}), +, \times, 1, 0)\) is a semiring, and \(E\) be a \(M\)-monadic proper expression over an alphabet \( \Sigma \).
    The \emph{series} \(S(E)\) \emph{associated} with \(E\) is inductively defined as follows:
    \begin{gather*}
        \begin{aligned}
            S(\varepsilon)(w)                  & =
            \begin{cases}
                1 & \text{if } w = \varepsilon, \\
                0 & \text{otherwise},
            \end{cases} &
            S(\emptyset)(w)                    & = 0, &
            S(a)(w)                            & =
            \begin{cases}
                1 & \text{if } w = a, \\
                0 & \text{otherwise},
            \end{cases}
        \end{aligned}\\
    \end{gather*}
    \begin{gather*}
        \begin{aligned}
            S(E_1 + E_2)     & = S(E_1) + S(E_2),      &
            S(E_1 \cdot E_2) & = S(E_1) \times S(E_2), &
            S(E_1^*)         & = {(S(E_1))}^\star,
        \end{aligned}\\
        \begin{aligned}
            S(\alpha \odot E_1)(w) & = \alpha \times S(E_1)(w), &
            S(E_1 \odot \alpha)(w) & = S(E_1)(w) \times \alpha,
        \end{aligned}\\
        S(f(E_1, \ldots, E_n)) = f(S(E_1), \ldots, S(E_n)),
    \end{gather*}
    where \(a\) is a symbol in \(\Sigma \), \((E_1, \ldots, E_n)\) are n \(M\)-monadic expressions over \( \Sigma \), \( \alpha \) is an element of \(M(\mathbbm{1})\) and \(f\) is a function from \({(M(\mathbbm{1}))}^n\) to \(M(\mathbbm{1})\).
\end{definition}
%
%
From now on, 
we consider the set \(\mathrm{Exp}(\Sigma)\) of \(M\)-monadic expressions over \(\Sigma \) to be endowed with the structure of a semiring, and two expressions denoting the same series to be equal.
\noindent The \emph{weight associated with} a word \(w\) in \(\Sigma^*\) \emph{by} \(E\) is the value \(\mathtt{weight}_w(E) = S(E)(w)\).
%
The nullability of a proper expression is the weight it associates with \(\varepsilon \),
following Definition~\ref{def partial nullability} and Definition~\ref{def series}.
\begin{proposition}
    Let \(M\) be a monad such that the structure \((M(\mathbbm{1}), +, \times, 1, 0)\) is a semiring.
    Let \(E\) be an \(M\)-monadic proper expression over \( \Sigma \).
    Then:
    \begin{equation*}
        \mathtt{Null}(E) = \mathtt{weight}_\varepsilon(E).
    \end{equation*}
\end{proposition}
The previous proposition implies that the weight of the empty word can be syntactically computed (\emph{i.e.} inductively computed from a monadic expression).
Now, let us show how to extend this computation by defining the computation of derivatives for monadic expressions.

\section{Monadic Supports for Expressions}\label{sec supp}

\noindent A \(\mathbb{K}\)\emph{-left-semimodule}, for a semiring \(\mathbb{K}=(K, \times, +, 1, 0)\), is a commutative monoid \((S, \pm, \underline{0})\) endowed with a function \(\triangleright \) from \(K \times S\) to \(S\) such that:
\begin{gather*}
    \begin{aligned}
        (k \times k') \triangleright s & = k \triangleright (k' \triangleright s),     &
        (k + k') \triangleright s      & = k \triangleright s \pm k' \triangleright s,
    \end{aligned}\\
    \begin{aligned}
        k \triangleright (s \pm s') & = k \triangleright s \pm k \triangleright s',     &
        1 \triangleright s          & = s,                                              &
        0 \triangleright s          & = k \triangleright \underline{0} = \underline{0}.
    \end{aligned}\\
\end{gather*}
A \(\mathbb{K}\)\emph{-right-semimodule} can be defined symmetrically.

An \emph{operad}~\cite{LV12,May06} is a structure \((O, {(\circ_{j})}_{j\in\mathbb{N}}, \mathrm{id})\) where
\(O\) is a graded set (\emph{i.e.} \(O = \bigcup_{n\in\mathbb{N}} O_n\)),
\(\mathrm{id}\) is an element of \(O_1\),
\(\circ_j\) is a function defined for any three integers \( (i, j, k) \)\footnote{every couple \( (i,k) \) unambiguously defines the domain and codomain of a function \( \circ_j \)} with \( 0<j\leq k \) in \( O_k\times O_{i} \rightarrow O_{k+i-1} \)
such that for any elements \( p_1 \in O_m \), \( p_2 \in O_n \), \( p_3 \in O_p \):
\begin{gather*}
    \forall 0 < j \leq m, \mathrm{id} \circ_1 p_1 = p_1 \circ_j \mathrm{id} =p_1,                                 \\
    \forall 0< j \leq m, 0< j' \leq n, p_1 \circ_j (p_2 \circ_{j'} p_3) = (p_1 \circ_j p_2) \circ_{j+{j'}-1} p_3, \\
    \forall 0 < {j'} \leq j \leq m, (p_1 \circ_j p_2) \circ_{j'} p_3 = (p_1 \circ_{j'} p_3) \circ_{j+p-1} p_2.
\end{gather*}
Combining these compositions \( \circ_j \), one can define a composition \( \circ \) sending \( O_k \times O_{i_1}\times \cdots\times O_{i_k} \) to \( O_{i_1+\cdots+i_k} \):
for any element \( (p, q_{1}, \ldots, q_{k}) \) in \( O_k \times O^k \),
\begin{equation*}
    p \circ (q_{1},\ldots,q_{k})= (\cdots((p \circ_k q_{k})\circ_{k-1} q_{{k-1}}\cdots)\cdots)\circ_1 q_{1}.
\end{equation*}
Conversely, the composition \( \circ \) can define the compositions \( \circ_j \) using the identity element:
for any two elements \( (p,q) \) in \( O_k\times O_i \), for any integer \( 0<j\leq k \):
\begin{equation*}
    p\circ_j q = p\circ (\underbrace{\mathrm{id},\ldots,\mathrm{id}}_{j-1\text{ times}},q,\underbrace{\mathrm{id},\ldots,\mathrm{id}}_{k-j\text{ times}}).
\end{equation*}
As an example, the set of \(n\)-ary functions over a set, with the identity function as unit, forms an operad.

A \emph{module over an operad} \((O, \circ, \mathrm{id})\) is a set \(S\) endowed with a function \(\divideontimes \) from \(O_n \times S^n\) to \(S\) such that
\begin{multline*}
    f \divideontimes (f_1 \divideontimes (s_{1, 1}, \ldots, s_{1, i_1}), \ldots, f_n \divideontimes (s_{n, 1}, \ldots, s_{n, i_n}))\\
    =
    (f \circ (f_1, \ldots, f_n)) \divideontimes (s_{1, 1}, \ldots, s_{1, i_1}, \ldots, s_{n, 1}, \ldots, s_{n, i_n}).
\end{multline*}

The extension of the computation of derivatives could be performed for any monad.
Indeed, any monad could be used to define well-typed auxiliary functions that mimic the classical computations.
However, some properties should be satisfied in order to compute weights equivalently to Definition~\ref{def series}.
Therefore, in the following we consider a restricted kind of monads.

A \emph{monadic support} is a structure \((M, +, \times, 1, 0, \pm, \underline{0}, \ltimes, \triangleright, \triangleleft,  \divideontimes)\) satisfying:
\begin{itemize}
    \item \(M\) is a monad,
    \item \(\mathbb{R} = (M(\mathbbm{1}), +, \times, 1, 0)\) is a semiring,
    \item \(\mathbb{M} = (M(\mathrm{Exp}(\Sigma)), \pm, \underline{0}) \) is a monoid,
    \item \((\mathbb{M}, \ltimes)\) is a \(\mathrm{Exp}(\Sigma)\)-right-semimodule,
    \item \((\mathbb{M}, \triangleright)\) is a \(\mathbb{R}\)-left-semimodule,
    \item \((\mathbb{M}, \triangleleft)\) is a \(\mathbb{R}\)-right-semimodule,
    \item \((M(\mathrm{Exp}(\Sigma)), \divideontimes)\) is a module for the operad of the functions over \(M(\mathbbm{1})\).
\end{itemize}
An \emph{expressive support} is a monadic support \((M, +, \times, 1, 0, \pm, \underline{0}, \ltimes, \triangleright,\triangleleft, \divideontimes)\) endowed with a function \(\mathtt{toExp}\) from \(M(\mathrm{Exp}(\Sigma))\) to \(\mathrm{Exp}(\Sigma)\) satisfying the following conditions:
\begin{align}
    \mathtt{weight}_w(\mathtt{toExp}(m))              & = m  \bind \mathtt{weight}_w \label{eq toExp weight}                             \\
    \mathtt{toExp}(m \ltimes F)                       & = \mathtt{toExp}(m) \cdot F,                                                     \\
    \mathtt{toExp}(m \pm m')                          & = \mathtt{toExp}(m) + \mathtt{toExp}(m'),                                        \\
    \mathtt{toExp}(m \triangleright x)                & = \mathtt{toExp}(m) \odot x,                                                     \\
    \mathtt{toExp}(x \triangleleft m)                 & = x \odot \mathtt{toExp}(m),                                                     \\
    \mathtt{toExp}(f\divideontimes(m_1, \ldots, m_n)) & = f(\mathtt{toExp}(m_1), \ldots, \mathtt{toExp}(m_n)). \label{eq toExp function}
\end{align}
Let us now illustrate this notion with three expressive supports that will allow us to model well-known derivatives computations.

\begin{example}[The \(\mathtt{Maybe}\) support]
    \begin{gather*}
        \begin{aligned}
            \mathtt{toExp}(\mathtt{Nothing}) & = 0, &
            \mathtt{toExp}(\mathtt{Just}(E)) & = E,
        \end{aligned}\\
        \begin{aligned}
            \mathtt{Nothing} + m                      & = m ,                  \\
            m + \mathtt{Nothing}                      & = m ,                  \\
            \mathtt{Just}(\top) + \mathtt{Just}(\top) & = \mathtt{Just}(\top),
        \end{aligned}
        \qquad \qquad
        \begin{aligned}
            \mathtt{Nothing} \times   m                      & =  \mathtt{Nothing},    \\
            m \times   \mathtt{Nothing}                      & =  \mathtt{Nothing},    \\
            \mathtt{Just}(\top) \times   \mathtt{Just}(\top) & =  \mathtt{Just}(\top),
        \end{aligned}\\
        \begin{aligned}
            \mathtt{Nothing} \pm m                 & = m,                     &
            m \pm \mathtt{Nothing}                 & = m,                     &
            \mathtt{Just}(E) \pm \mathtt{Just}(E') & = \mathtt{Just}(E + E'),
        \end{aligned}\\
        \begin{aligned}
            1             & = \mathtt{Just}(\top), &
            0             & = \mathtt{Nothing},    &
            \underline{0} & = \mathtt{Nothing},
        \end{aligned}\\
        \begin{aligned}
            m \ltimes F & = (\lambda E \rightarrow E\cdot F) \dollarfmap m,
        \end{aligned}\\
        \begin{aligned}
            m \triangleright m' & = m \bind (\lambda x \rightarrow m'), &
            m \triangleleft m'  & = m' \bind (\lambda x \rightarrow m),
        \end{aligned}\\
        \begin{aligned}
            f \divideontimes (m_1, \ldots, m_n) & = \mathtt{pure}(f(\mathtt{toExp}(m_1), \ldots, \mathtt{toExp}(m_n))).
        \end{aligned}
    \end{gather*}
\end{example}
\begin{example}[The \(\mathtt{Set}\) support]
    \begin{gather*}
        \begin{aligned}
            \mathtt{toExp}(\{E_1, \ldots, E_n\}) & = E_1 + \cdots + E_n,
        \end{aligned}\\
        \begin{aligned}
            +             & = \cup,      &
            \times        & = \cap,      &
            \pm           & = \cup,      &
            1             & = \{\top \}, &
            0             & = \emptyset, &
            \underline{0} & = \emptyset,
        \end{aligned}\\
        \begin{aligned}
            m \ltimes F & = (\lambda E \rightarrow E\cdot F) \dollarfmap m,
        \end{aligned}\\
        \begin{aligned}
            m \triangleright m' & = m \bind (\lambda x \rightarrow m'), &
            m \triangleleft m'  & = m' \bind (\lambda x \rightarrow m),
        \end{aligned}\\
        \begin{aligned}
            f \divideontimes (m_1, \ldots, m_n) & = \mathtt{pure}(f(\mathtt{toExp}(m_1), \ldots, \mathtt{toExp}(m_n))).
        \end{aligned}
    \end{gather*}
\end{example}
\begin{example}[The \(\mathtt{LinComb}(\mathbb{K})\) support]
    \begin{gather*}
        \begin{aligned}
            \mathtt{toExp}((k_1, E_1) \boxplus \cdots \boxplus (k_n, E_n)) & = k_1 \odot E_1 + \cdots + k_n \odot E_n,
        \end{aligned}\\
        \begin{aligned}
            +                           & = \boxplus,            &
            (k, \top) \times (k', \top) & = (k \times k', \top), &
            1                           & = (1, \top),           &
            0                           & = (0, \top),
        \end{aligned}\\
        \begin{aligned}
            \pm           & = \boxplus,  &
            \underline{0} & = (0, \top),
        \end{aligned}\\
        \begin{aligned}
            m \ltimes F & = (\lambda E \rightarrow E\cdot F) \dollarfmap m,
        \end{aligned}\\
        \begin{aligned}
            m \triangleright m' & = m \bind (\lambda x \rightarrow m'),             &
            m \triangleleft k   & = (\lambda E \rightarrow E\odot k) \dollarfmap m,
        \end{aligned}\\
        \begin{aligned}
            f \divideontimes (m_1, \ldots, m_n) & = \mathtt{pure}(f(\mathtt{toExp}(m_1), \ldots, \mathtt{toExp}(m_n))).
        \end{aligned}
    \end{gather*}
\end{example}

\section{Monadic Derivatives}\label{sec deriv}

In the following, \((M, +, \times, 1, 0, \pm, \underline{0}, \ltimes, \triangleright,\triangleleft, \mathtt{toExp})\) is an expressive support.

\begin{definition}\label{def der symb}
    The \emph{derivative} of an \(M\)-monadic expression \(E\) over \( \Sigma \) w.r.t.\ a symbol \(a\) in \(\Sigma \) is the element \( d_a(E) \) in \( M (\mathrm{Exp}(\Sigma)) \) inductively defined as follows:
    \begin{gather*}
        \begin{aligned}
            d_a(\varepsilon) & = \underline{0}, &
            d_a(\emptyset)   & = \underline{0}, &
            d_a(b)           & =
            \begin{cases}
                \mathtt{pure}(\varepsilon) & \text{if } a = b, \\
                \underline{0}              & \text{otherwise,}
            \end{cases}
        \end{aligned}\\
        \begin{aligned}
            d_a(E_1 + E_2) & = d_a(E_1) \pm d_a(E_2),  &
            d_a(E_1^*)     & = d_a(E_1) \ltimes E_1^*,
        \end{aligned}\\
        \begin{aligned}
            d_a(E_1 \cdot E_2) & =
            d_a(E_1) \ltimes E_2
            \pm \mathtt{Null}(E_1) \triangleright d_a(E_2),
        \end{aligned}\\
        \begin{aligned}
            d_a(\alpha \odot E_1) & = \alpha \triangleright d_a(E_1), &
            d_a(E_1 \odot \alpha) & = d_a(E_1) \triangleleft \alpha,
        \end{aligned}\\
        d_a(f(E_1, \ldots, E_n)) = f \divideontimes (d_a(E_1), \ldots, d_a(E_n))
    \end{gather*}
    where \(b\) is a symbol in \(\Sigma \), \((E_1, \ldots, E_n)\) are n \(M\)-monadic expressions over \( \Sigma \), \( \alpha \) is an element of \(M(\mathbbm{1})\) and \(f\) is a function from \({(M(\mathbbm{1}))}^n\) to \(M(\mathbbm{1})\).
\end{definition}
The link between derivatives and series can be stated as follows, which is an alternative description of the classical quotient.
\begin{proposition}\label{prop der and quot}
    Let \(E\) be an \(M\)-monadic expression over an alphabet \(\Sigma \), \(a\) be a symbol in \(\Sigma \) and \(w\) be a word in \(\Sigma^*\).
    Then:
    \begin{equation*}
        \mathtt{weight}_{aw} (E) = d_a(E) \bind \mathtt{weight}_w.
    \end{equation*}
\end{proposition}
\begin{proof}
    Let us proceed by induction over the structure of \(E\).
    All the classical cases (\emph{i.e.} the function operator left aside) can be proved following the classical methods (\cite{Ant96,Brzo64,LS05}).
    Therefore, let us consider this last case.
    \begin{align*}
         & d_a(f(E_1, \ldots, E_n)) \bind \mathtt{weight}_w                                                                                                                             \\
         & \qquad = \mathtt{weight}_w (\mathtt{toExp} (d_a(f(E_1, \ldots, E_n))))                                        & (\text{Eq}~\eqref{eq toExp weight})                          \\
         & \qquad = \mathtt{weight}_w (\mathtt{toExp} (f \divideontimes (d_a(E_1), \ldots, d_a(E_n)))                    & (\text{Def}~\ref{def der symb}))                             \\
         & \qquad = \mathtt{weight}_w (f (\mathtt{toExp}(d_a(E_1)), \ldots, \mathtt{toExp}(d_a(E_n))))                   & (\text{Eq}~\eqref{eq toExp function})                        \\
         & \qquad = f (\mathtt{weight}_w(\mathtt{toExp}(d_a(E_1))), \ldots, \mathtt{weight}_w(\mathtt{toExp}(d_a(E_n)))) & (\text{Def}~\ref{def series}, \text{Eq}~\eqref{eq f series}) \\
         & \qquad = f (d_a(E_1) \bind \mathtt{weight}_w, \ldots, d_a(E_n) \bind \mathtt{weight}_w)                       & (\text{Eq}~\eqref{eq toExp weight})                          \\
         & \qquad = f (\mathtt{weight}_{aw} (E_1), \ldots, \mathtt{weight}_{aw} (E_n))                                   & (\text{Ind.\ hyp.})                                          \\
         & \qquad = \mathtt{weight}_{aw} ( f(E_1, \ldots, E_n))                                                          & (\text{Def}~\ref{def series}, \text{Eq}~\eqref{eq f series})
    \end{align*}
    \qed{}
\end{proof}
Let us define how to extend the derivative computation from symbols to words, using the monadic functions.
\begin{definition}
    The \emph{derivative} of an \(M\)-monadic expression \(E\) over \( \Sigma \) w.r.t.\ a word \(w\) in \(\Sigma^*\) is the element \( d_w(E) \) in \( M (\mathrm{Exp}(\Sigma)) \) inductively defined as follows:
    \begin{align*}
        d_\varepsilon(E) & = \mathtt{pure}(E), & d_{a\cdot v}(E) & = d_a(E) \bind d_v,
    \end{align*}
    where \(a\) is a symbol in \( \Sigma \) and \(v\) a word in \(\Sigma^*\).
\end{definition}

Finally, it can be easily shown, by induction over the length of the words, following Proposition~\ref{prop der and quot}, that the derivatives computation can be used to define a syntactical computation of the weight of a word associated with an expression.
\begin{theorem}
    Let \(E\) be an \(M\)-monadic expression over an alphabet \( \Sigma \) and \(w\) be a word in \(\Sigma^*\).
    Then:
    \begin{equation*}
        \mathtt{weight}_w(E) = d_w(E) \bind \mathtt{Null}.
    \end{equation*}
\end{theorem}

Notice that, restraining monadic expressions to regular ones,
\begin{itemize}
    \item the \(\mathtt{Maybe}\) support leads to the classical derivatives~\cite{Brzo64},
    \item the \(\mathtt{Set}\) support leads to the partial derivatives~\cite{Ant96},
    \item the \(\mathtt{LinComb}\) support leads to the derivatives with multiplicities~\cite{LS05}.
\end{itemize}

\begin{example}\label{ex: calc der}
    Let us consider the function \(\mathtt{ExtDist}\) defined in Example~\ref{ex:def extdist} and the \(\mathtt{LinComb}(\mathbb{N})\)-monadic expression \(E = \mathtt{ExtDist}(a^*b^* + b^*a^*, b^*a^*b^*, a^*b^*a^*)\).
    \begin{align*}
        d_a(E)                   & = \mathtt{ExtDist}(a^*b^* + a^*, a^*b^*, a^*b^*a^*+ a^*)          \\
        d_{aa}(E)                & = \mathtt{ExtDist}(a^*b^* + a^*, a^*b^*, a^*b^*a^* + 2 \odot a^*) \\
        d_{aaa}(E)               & = \mathtt{ExtDist}(a^*b^* + a^*, a^*b^*, a^*b^*a^* + 3 \odot a^*) \\
        d_{aab}(E)               & = \mathtt{ExtDist}(b^*, b^*, b^*a^*)                              \\
        \mathtt{weight}_{aaa}(E) & = d_{aaa}(E) \bind \mathtt{Null}                                  \\
                                 & = \mathtt{ExtDist}(1 + 1, 1, 1 + 3) = 4 - 1 = 3                   \\
        \mathtt{weight}_{aab}(E) & = d_{aab}(E) \bind \mathtt{Null}
        = \mathtt{ExtDist}(1,1,1) = 0
    \end{align*}
\end{example}

In the next section, we show how to compute the derivative automaton associated with an expression.
\section{Automata Construction}\label{sec aut cons}

A \emph{category} \( \mathcal{C}\) is defined by:
\begin{itemize}
  \item a class \(\mathrm{Obj}_{\mathcal{C}}\) of \emph{objects},
  \item for any two objects \(A\) and \(B\), a set \( \mathrm{Hom}_{\mathcal{C}}(A,B)\) of \emph{morphisms},
  \item for any three objects \(A\), \(B\) and \(C\), an associative \emph{composition function} \(\circ_{\mathcal{C}} \) in \( \mathrm{Hom}_{\mathcal{C}}(B,C) \longrightarrow  \mathrm{Hom}_{\mathcal{C}}(A,B) \longrightarrow \mathrm{Hom}_{\mathcal{C}}(A,C)\),
  \item for any object \(A\), an \emph{identity morphism} \(\mathrm{id}_A \) in \( \mathrm{Hom}_{\mathcal{C}}(A,A) \), such that for any morphisms \(f\) in \(\mathrm{Hom}_{\mathcal{C}}(A,B)\) and \(g\) in \( \mathrm{Hom}_{\mathcal{C}}(B,A)\), \( f \circ_{\mathcal{C}} \mathrm{id}_A = f\) and \( \mathrm{id}_A \circ_{\mathcal{C}} g = g \).
\end{itemize}

Given a category \(\mathcal{C}\), a \(\mathcal{C}\)-automaton is a tuple \((\Sigma, I, Q, F, i, \delta, f)\) where
\begin{itemize}
  \item \(\Sigma \) is a set of symbols (the alphabet),
  \item \(I\) is the initial object, in \(\mathrm{Obj}(\mathcal{C})\),
  \item \(Q\) is the state object, in \(\mathrm{Obj}(\mathcal{C})\),
  \item \(F\) is the final object, in \(\mathrm{Obj}(\mathcal{C})\),
  \item \(i\) is the initial morphism, in \(\mathrm{Hom}_{\mathcal{C}}(I,Q)\),
  \item \(\delta \) is the transition function, in \(\Sigma \longrightarrow \mathrm{Hom}_{\mathcal{C}}(Q,Q)\),
  \item \(f\) is the value morphism, in \(\mathrm{Hom}_{\mathcal{C}}(Q,F)\).
\end{itemize}

The function \(\delta \) can be extended as a monoid morphism from the free monoid \((\Sigma^*, \cdot, \varepsilon)\)
to the morphism monoid \((\mathrm{Hom}_{\mathcal{C}}(Q,Q), \circ_{\mathcal{C}}, \mathrm{id}_Q)\), leading to the following weight definition.

The weight associated by a \(\mathcal{C}\)-automaton \(A=(\Sigma, I, Q, F, i, \delta, f)\) with a word \(w\) in \(\Sigma^*\)
is the morphism \(\mathtt{weight}(w)\) in \(\mathrm{Hom}_{\mathcal{C}}(I,F)\) defined by
\begin{equation*}
  \mathtt{weight}(w) = f \circ_\mathcal{C} \delta(w) \circ_\mathcal{C} i.
\end{equation*}

If the ambient category is the category of sets, and if \(I = \mathbbm{1} \), the weight of a word is equivalently
an element of \(F\).
Consequently, a deterministic (complete) automaton is equivalently a Set-automaton with \(\mathbbm{1}\) as the initial object
and \(\mathbb{B}\) as the final object.

Given a monad \(M\), the Kleisli composition of two morphisms \(f \in \mathrm{Hom}_{\mathcal{C}}(A, B) \) and \(g \in \mathrm{Hom}_{\mathcal{C}}(B, C)\)
is the morphism \((f \kleisli g) (x) = f(x) \bind g \) in \( \mathrm{Hom}_{\mathcal{C}}(A, C) \).
This composition defines a category, called the Kleisli category \(\mathcal{K}(M)\) of \(M\), where:
\begin{itemize}
  \item the objects are the sets,
  \item the morphisms between two sets \(A\) and \(B\) are
  the functions between \(A\) and \(M(B)\),
  \item the identity is the function \( \mathtt{pure} \).
\end{itemize}

Considering these categories:
\begin{itemize}
  \item a deterministic automaton is equivalently a \(\mathcal{K}(\mathtt{Maybe})\)-automaton,
  \item a nondeterministic automaton is equivalently a \(\mathcal{K}(\mathtt{Set})\)-automaton,
  \item a weighted automaton over a semiring \(\mathbb{K}\) is equivalently a \(\mathcal{K}(\mathtt{LinComb}(\mathbb{K}))\)-automaton,
\end{itemize}
all with \(\mathbbm{1}\) as both the initial object and the final object.

Furthermore, for a given expression \(E\), if \(i = \mathtt{pure}(E)\), \(\delta(a)(E') = \mathrm{d}_a(E')\) and \(f = \mathtt{Null}\),
we can compute the well-known derivative automata using the three previously defined supports, and the accessible part of these automata are finite ones
as far as classical expressions are concerned~\cite{Brzo64,Ant96,LS05}.

More precisely, extended expressions can lead to infinite automata, as shown in the next example.
\begin{example}
  Considering the computations of Example~\ref{ex: calc der},
  it can be shown that
  \begin{equation*}
    d_{a^n}(E) = \mathtt{ExtDist}(a^*b^*+a^*, a^*b^*, a^*b^*a^*+ n \odot a^*).
  \end{equation*}
  Hence, there is not a finite number of derivated terms, that are the states in the classical derivative automaton.
  This infinite automaton is represented in Figure~\ref{ref: fig auto lin comb monad}, where the final weights of the states are represented by double edges.
  The sink states are omitted.
\end{example}

\begin{figure}[H]
  \centering
  \begin{tikzpicture}[node distance=3cm,bend angle=30,transform shape,scale=1]
    \node[state, rounded rectangle] (1)  {\(\mathtt{ExtDist}(a^*b^* + b^*a^*, b^*a^*b^*, a^*b^*a^*)\)} ;

    \node[state, rounded rectangle, below left of=1,node distance = 3cm] (2) {\( \mathtt{ExtDist}(a^*b^* + a^*,a^*b^*, a^*b^*a^*
      +a^*)  \)};

    \node[state, rounded rectangle,below left of=2,node distance = 2.5cm] (2l) {\(  \mathtt{ExtDist}(a^*b^* + a^*,a^*b^*, a^*b^*a^*
      +2 \odot a^*) \)};

    \node[state, rounded rectangle,right of=2l,node distance = 4.5cm] (2r) {\(  \mathtt{ExtDist}(b^*,b^*,b^*a^*) \)};

    \node[state, rounded rectangle,below  of=2r,node distance = 3.5cm] (2+r) {\(  \mathtt{ExtDist}(0,0,a^*)  \)};

    \node[state, rounded rectangle,  right of=2,node distance=5cm] (3) {\( \mathtt{ExtDist}(b^*+b^*a^*,b^*a^*b^*+b^*,b^*a^*) \)};

    \node[state, rounded rectangle, below right of=3,node distance=3.5cm] (3l) {\( \mathtt{ExtDist}(b^*+b^*a^*,b^*a^*b^*+2\odot b^*,b^*a^*) \)};
    \node[state, rounded rectangle, below right of=2r,node distance=1.8cm] (3r) {\( \mathtt{ExtDist}(a^*,a^*b^*,a^*) \)};

    \node[ state, rounded rectangle, below  of=3r,node distance=1.5cm] (3ll) {\(  \mathtt{ExtDist}(0,b^*,0)  \)};

    \node[state, rounded rectangle,below  of=2l,node distance = 2.5cm] (2ll) {\(  \mathtt{ExtDist}(a^*b^* + a^*,a^*b^*, a^*b^*a^*
      +n \odot a^*) \)};


    \node[state, rounded rectangle,below  of=3l,node distance = 3.5cm] (3lll) {\(  \mathtt{ExtDist}(b^*+b^*a^*,b^*a^*b^*+n \odot b^*,b^*a^*) \)};

    \draw (1) ++(0cm,1cm) node {}  edge[->] (1);
    \draw[right] (1) ++(-4cm,0cm) node {$1$}  edge[double, Implies-] (1);
    \draw[right] (2) ++(-4cm,0cm) node {$1$}  edge[double, Implies-] (2);
    \draw[right] (2l) ++(-1cm,1cm) node {$2$}  edge[double, Implies-] (2l);
    \draw[right] (2ll) ++(-1cm,1cm) node {$n$}  edge[double, Implies-] (2ll);
    \draw[right] (2+r) ++(-2cm,0cm) node {$1$}  edge[double, Implies-] (2+r);
    \draw[right] (3) ++(4cm,0cm) node {$1$}  edge[double, Implies-] (3);
    \draw[right] (3l) ++(1cm,1cm) node {$2$}  edge[double, Implies-] (3l);
    \draw[right] (3ll) ++(-1cm,1cm) node {$1$}  edge[double, Implies-] (3ll);
    \draw[right] (3lll) ++(1cm,1cm) node {$n$}  edge[double, Implies-] (3lll);
    \draw[right] (2ll) ++(0cm,-1cm) node {}  edge[dotted, -] (2ll);
    \draw[right] (3lll) ++(0cm,-1cm) node {}  edge[dotted, -] (3lll);

    \path[->]
    (1) edge[->, left,] node {\( a\)} (2)
    (1) edge[->, right] node {\( b \)} (3)
    (2) edge[->, left] node {\( b\)} (2r)
    (2) edge[->, left] node {\( a \)} (2l)
    (2l) edge[->, above] node {\( b \)} (2r)
    (3) edge[->, above right,bend right=5] node {\(
        a \)} (3r)
    (3) edge[->, right,bend right=-20] node {\( b \)} (3l)
    (3l) edge[->, above, in = 0, out = 235] node {\(a \)} (3r)
    (3r) edge[->, right] node {\(b \)} (3ll)
    (2l) edge[-, left] node {} (-3.9,-4.7)

    (-3.9,-5.5)  edge[->, right,] node {$a$} (2ll)
    (2ll) edge[->, above left] node {\(b \)} (2r)
    (3l) edge[-, left] node {} (5.35,-5.5)
    (2r) edge[->, bend right=30, left] node {\( a \)} (2+r)
    (5.35,-6.5) edge[->, left] node {$b$} (3lll)

    (3lll) edge[->, left] node {\(a \)} (3r)
    ;
    \draw[dashed] (-3.9,-4.7)  to (-3.9,-5.5) ;
    \draw[dashed] (5.35,-5.5)   to (5.35,-6.5) ;

    \path[->,right] (2r) edge[loop above,->,above] node {$b$} ();
    \path[->,right] (3r) edge[out=145, in=115, loop,above] node {$a$} ();
    \path[->] (3ll) edge[out=-45, in=-25, loop,below] node {$b$} (3ll);
    \path[->,right] (2+r) edge[loop  below,->,below] node {$a$} ();

  \end{tikzpicture}
  \caption{The (infinite) derivative weighted automaton associated with \(E\).}%
  \label{ref: fig auto lin comb monad}
\end{figure}
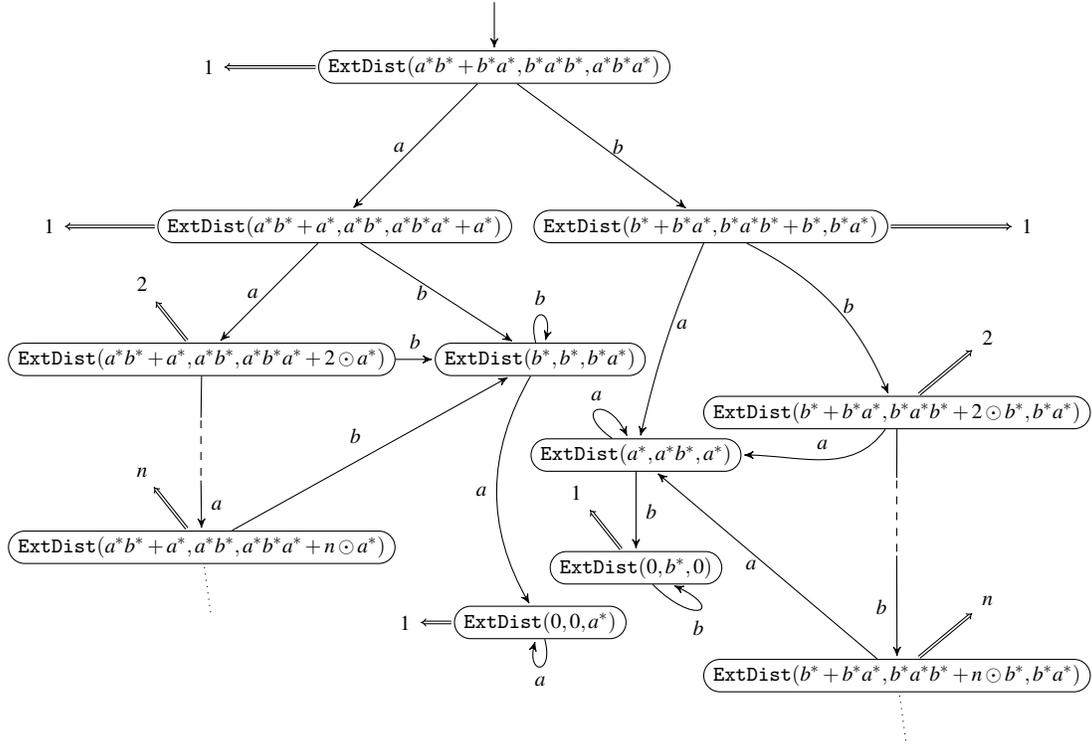

In the following section, let us show how to model a new monad in order to solve this problem.

\section{The Graded Module Monad}\label{sec new mon}

Let us consider an operad \(\mathbb{O}=(O,\circ,\mathrm{id})\) and the \emph{association} sending:
\begin{itemize}
    \item any set \(S\) to \(\bigcup_{n\in\mathbb{N}} O_n \times S^n\),
    \item any \(f\) in \(S \rightarrow S'\) to the function \(g\) in \( \bigcup_{n\in\mathbb{N}} O_n \times S^n \rightarrow \bigcup_{n\in\mathbb{N}} O_n \times S'^n\):
          \begin{equation*}
              g(o, (s_1, \ldots, s_n)) = (o, (f(s_1), \ldots, f(s_n)))
          \end{equation*}
\end{itemize}
It can be checked that this is a functor, denoted  by \(\mathtt{GradMod}(\mathbb{O})\).
Moreover, it forms a monad considering the two following functions:
\begin{align*}
    \mathtt{pure}(s)                & = (\mathrm{id}, s),                                                                                \\
    (o, (s_1, \ldots, s_n)) \bind f & = (o \circ (o_1, \ldots, o_n), (s_{1,1}, \ldots, s_{1, i_1}, \ldots, s_{n,1}, \ldots, s_{n, i_n}))
\end{align*}
where \(f(s_j) = (o_j, s_{j, 1}, \ldots, s_{j, i_j})\).
However, notice that \(\mathtt{GradMod}(\mathbb{O})(\mathbbm{1})\) cannot be easily evaluated as a value space.
Thus, let us compose it with another monad.
As an example, let us consider a semiring \(\mathbb{K}=(K,\times,+,1,0)\) and the operad \(\mathbb{O}\) of the \(n\)-ary functions over \(K\).
Hence, let us define the functor\footnote{it is folk knowledge that the composition of two functors is a functor.} \(\mathtt{GradComb}(\mathbb{O},\mathbb{K})\) that sends \(S\) to \(\mathtt{GradMod}(\mathbb{O})(\mathtt{LinComb}(\mathbb{K})(S))\).

To show that this combination is a monad, let us first define a function \(\alpha \) sending \(\mathtt{GradComb}(\mathbb{O},\mathbb{K})(S)\) to \(\mathtt{GradMod}(\mathbb{O})(S)\).
It can be easily done by converting a linear combination into an \emph{operadic combination}, \emph{i.e.} an element in \(\mathtt{GradMod}(\mathbb{O})(S)\), with the following function \(\mathtt{toOp}\):
\begin{gather*}
    \begin{multlined}
        \mathtt{toOp}((k_1, s_1) \boxplus \cdots \boxplus (k_n, s_n)) \\
        \qquad \qquad =
        (\lambda (x_1, \ldots, x_n) \rightarrow k_1 \times x_1 + \cdots + k_n \times x_n, (s_1, \ldots, s_n)),
    \end{multlined}\\
    \alpha(o, (\mathcal{L}_1, \ldots, \mathcal{L}_n))
    =
    (o\circ (o_1, \ldots, o_n), (s_{1,1}, \ldots, s_{1, i_1}, \ldots, s_{n,1}, \ldots, s_{n, i_n}))
\end{gather*}
where \(\mathtt{toOp}(\mathcal{L}_j) = (o_j, (s_{j,1}, \ldots, s_{j, i_j}))\).

Consequently, we can define the monadic functions as follows:
\begin{align*}
    \mathtt{pure}(s)                                    & = (\mathrm{id}, (1, s)),                                    \\
    (o, (\mathcal{L}_1, \ldots, \mathcal{L}_n)) \bind f & = \alpha(o, (\mathcal{L}_1, \ldots, \mathcal{L}_n)) \bind f
\end{align*}
where the second occurrence of \(\bind \) is the monadic function associated with the monad \(\mathtt{GradMod}(\mathbb{O})\).

Let us finally define an expressive support for this monad:
\begin{gather*}
    \begin{aligned}
        \mathtt{toExp}(o, (\mathcal{L}_1, \ldots, \mathcal{L}_n)) & = o(\mathtt{toExp}(\mathcal{L}_1), \ldots, \mathtt{toExp}(\mathcal{L}_n)),
    \end{aligned}\\
    \begin{aligned}
        (o, (\mathcal{L}_1, \ldots, \mathcal{L}_n)) + (o', (\mathcal{L}'_1, \ldots, \mathcal{L}'_{n'}))
         & = (o + o', (\mathcal{L}_1, \ldots, \mathcal{L}_n,\mathcal{L}'_1, \ldots, \mathcal{L}'_{n'}))      \\
        (o, (\mathcal{L}_1, \ldots, \mathcal{L}_n)) \times (o', (\mathcal{L}'_1, \ldots, \mathcal{L}'_{n'}))
         & = (o \times o', (\mathcal{L}_1, \ldots, \mathcal{L}_n,\mathcal{L}'_1, \ldots, \mathcal{L}'_{n'}))
    \end{aligned}\\
    \begin{aligned}
        \pm           & = +,                        &
        1             & = (\mathrm{id}, (1, \top)), &
        0             & = (\mathrm{id}, (0, \top)), &
        \underline{0} & = (\mathrm{id}, (0, \top)),
    \end{aligned}\\
    \begin{aligned}
        m \ltimes F & = \mathtt{pure}(\mathtt{toExp}(m) \cdot F),
    \end{aligned}\\
    \begin{aligned}
        (o, (\mathcal{M}_1, \ldots, \mathcal{M}_k)) \triangleright (o', (\mathcal{L}_1, \ldots, \mathcal{L}_{n})) & =
        (o(\mathcal{M}_1, \ldots, \mathcal{M}_k) \times o', (\mathcal{L}_1, \ldots, \mathcal{L}_{n})),                                                                                                                \\
        (o, (\mathcal{L}_1, \ldots, \mathcal{L}_n)) \triangleleft (o', (\mathcal{M}_1, \ldots, \mathcal{M}_{k}))  & = (o \times o'(\mathcal{M}_1, \ldots, \mathcal{M}_{k}), (\mathcal{L}_1, \ldots, \mathcal{L}_{n}))
    \end{aligned}\\
    \begin{multlined}
        f \divideontimes ((o_1, (\mathcal{L}_{1, 1}, \ldots, \mathcal{L}_{1, i_1})), \ldots, (o_n, (\mathcal{L}_{n, 1}, \ldots, \mathcal{L}_{n, i_n}))) \\
        \qquad \qquad = (f\circ(o_1, \ldots, o_n), (\mathcal{L}_{1, 1}, \ldots, \mathcal{L}_{1, i_1}, \ldots, \mathcal{L}_{n, 1}, \ldots, \mathcal{L}_{n, i_n}))
    \end{multlined}\\
    \begin{aligned}
        \text{where }(o + o')(x_1, \ldots, x_{n + n'}) & = o(x_1, \ldots, x_n) + o'(x_{n + 1}, \ldots, x_{n + n'})      \\
        (o \times o')(x_1, \ldots, x_{n + n'})         & = o(x_1, \ldots, x_n) \times o'(x_{n + 1}, \ldots, x_{n + n'})
    \end{aligned}
\end{gather*}
\begin{example}\label{ex finite deriv}
    Let us consider that two elements in \(\mathtt{GradComb}(\mathbb{O},\mathbb{K})(\mathrm{Exp}(\Sigma))\) are equal if they have the same image by \(\mathtt{toExp}\).
    Let us consider the expression \(E = \mathtt{ExtDist}(a^*b^* + b^*a^*, b^*a^*b^*, a^*b^*a^*)\) of Example~\ref{ex: calc der}.
    \begin{align*}
        d_a(E)                   & =    \mathtt{ExtDist} \divideontimes ((+, (a^*b^*, a^*)), (\mathrm{id}, a^*b^*), (+, (a^*b^*a^*, a^*)))           \\
                                 & = (\mathtt{ExtDist}\circ (+, \mathrm{id}, +), (a^*b^*, a^*, a^*b^*, a^*b^*a^*, a^*))                              \\
        d_{aa}(E)                & = (\mathtt{ExtDist}\circ (+, \mathrm{id}, +\circ(+, \mathrm{id})), (a^*b^*, a^*, a^*b^*, a^*b^*a^*, a^*, a^*))    \\
                                 & = (\mathtt{ExtDist}\circ (+, \mathrm{id}, + \circ(\mathrm{id}, 2 \times)), (a^*b^*, a^*, a^*b^*, a^*b^*a^*, a^*)) \\
        d_{aaa}(E)               & = (\mathtt{ExtDist}\circ (+, \mathrm{id}, +\circ(\mathrm{id}, 3 \times)), (a^*b^*, a^*, a^*b^*, a^*b^*a^*, a^*))  \\
        d_{aab}(E)               & = (\mathtt{ExtDist}\circ (+, \mathrm{id}, +), (b^*, \emptyset, b^*, b^*a^*, \emptyset))                           \\
                                 & = (\mathtt{ExtDist}, (b^*, b^*, b^*a^*))                                                                          \\
        \mathtt{weight}_{aaa}(E) & = d_{aaa}(E) \bind \mathtt{Null}                                                                                  \\
                                 & = \mathtt{ExtDist}\circ (+, \mathrm{id}, +) (1, 1, 1, 1, 3)                                                       \\
                                 & = \mathtt{ExtDist}(1 + 1, 1, 1 + 3) = 4 - 1 = 3                                                                   \\
        \mathtt{weight}_{aab}(E) & = d_{aab}(E) \bind \mathtt{Null}
        = \mathtt{ExtDist}(1,1,1) = 0
    \end{align*}
    Using this monad, the number of derivated terms, that is the number of states in the associated derivative automaton,
    is finite.
    Indeed, the computations are absorbed in the transition structure.
    This automaton is represented in Figure~\ref{ref: fig auto grad comb monad}.
    Notice that the dashed rectangle represent the functions that are composed during the traversal associated with a word.
    The final weights are represented by double edges.
    The sink states are omitted.
    The state \(b^*\) is duplicated to simplify the representation.
\end{example}

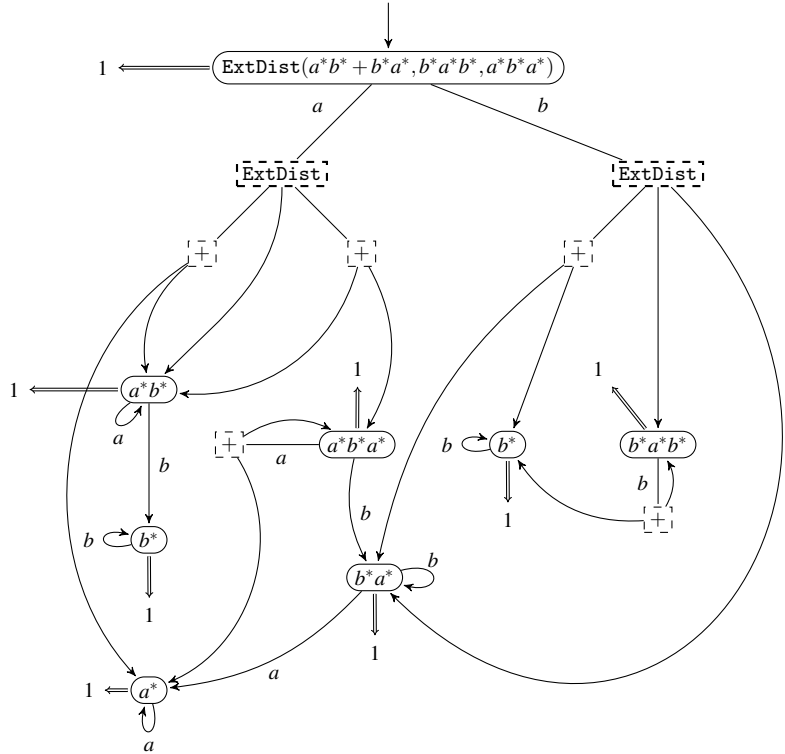
\begin{figure}[H]
    \centering
    \begin{tikzpicture}[node distance=2.5cm,bend angle=30,transform shape,scale=1]
        \node[state, rounded rectangle] (1)  {\(\mathtt{ExtDist}(a^*b^* + b^*a^*, b^*a^*b^*, a^*b^*a^*)\)} ;

        \node[thick, state,rectangle, dashed, below left of=1,node distance = 2cm] (2) {\( \mathtt{ExtDist}  \)};

        \node[state, rectangle, dashed, below left of=2,node distance = 1.5cm] (2+l) {\( +  \)};
        \node[state , rectangle, dashed,below right of=2,node distance = 1.5cm] (2+r) {\( +  \)};

        \node[thick,state, dashed, rectangle, right of=2,node distance=5cm] (3) {\( \mathtt{ExtDist} \)};

        \node[state ,  rectangle, dashed ,below left  of=3,node distance=1.5cm] (3+r) {\( + \)};

        \node[below of=2,node distance = 1.1cm, draw = none] (2id) {};
        \node[below right of=3,node distance=2.5cm, draw = none] (3+l) {};
        \node[ below  of=3,node distance=1.1cm, draw = none] (3id) {};

        \node[state, rounded rectangle, below  of=3id,node distance=2.5cm] (4) {\(b^*a^*b^* \)};

        \node[state, rectangle, dashed, below  of=4,node distance=1cm] (4+) {\( + \)};
        \node[state, rounded rectangle, left of=4,node distance=4cm] (5) {\(a^*b^*a^* \)};

        \node[ state, rectangle, dashed,  left  of=5,node distance=1.7cm] (5+) {\( + \)};

        \node[state, rounded rectangle, below left of=2id,node distance=2.5cm] (7) {\(a^*b^* \)};
        \node[state, rounded rectangle,below  of=7,node distance=4cm] (6) {\(a^* \)};
        \node[state, rounded rectangle, left of=4,node distance=2cm] (8) {\(b^* \)};

        \node[state, rounded rectangle,  below  left of=8,node distance=2.5cm] (9) {\( b^*a^* \)};

        \node[state, rounded rectangle,  above of=6,node distance=2cm] (81) {\(b^* \)};

        \draw (1) ++(0cm,1cm) node {}  edge[->] (1);
        \draw[right] (1) ++(-4cm,0cm) node {$1$}  edge[double, Implies-] (1);
        \draw[right] (6) ++(-1cm,0cm) node {$1$}  edge[double, Implies-] (6);
        \draw[right] (7) ++(-2cm,0cm) node {$1$}  edge[double, Implies-] (7);
        \draw[right] (9) ++(-0.2cm,-1cm) node {$1$}  edge[double, Implies-] (9);
        \draw[right] (5) ++(-0.2cm,1cm) node {$1$}  edge[double, Implies-] (5);
        \draw[right] (4) ++(-1cm,1cm) node {$1$}  edge[double, Implies-] (4);
        \draw[right] (8) ++(-0.2cm,-1cm) node {$1$}  edge[double, Implies-] (8);
        \draw[right] (81) ++(-0.2cm,-1cm) node {$1$}  edge[double, Implies-] (81);

        \path[->]
        (1) edge[-, above left] node {\( a\)} (2)
        (1) edge[-, above right] node {\( b \)} (3)

        (2) edge[->, out = 270, in = 45] node {} (7)

        (2) edge[-, right] node {} (2+r)
        (2) edge[-, right] node {} (2+l)
        (3) edge[-, right] node {} (3+r)

        (3) edge[->, ] node {} (4)

        (3) edge[->, out = -45, in = -45, looseness = 2] node {} (9)

        (3+r) edge[->] node {} (8)

        (3+r) edge[->, bend right=20] node {} (9)
        (2+r)  edge[->,bend right=-40] node {} (7)
        (2+r)  edge[->,bend right=-30] node {} (5)
        (2+l)  edge[->,bend right=50] node {} (6)
        (2+l)  edge[->,bend right=30] node {} (7)
        (5)  edge[-] node {\(a\)} (5+)
        (5+)  edge[->,bend right=-30] node {} (5)
        (5+)  edge[->,bend right=-50] node {} (6)
        (4)  edge[-,left] node {$b$} (4+)
        (4+)  edge[->,bend right=30] node {} (4)
        (4+)  edge[->,bend right=-30] node {} (8)
        (5)  edge[->,bend right=20, right] node {$b$} (9)
        (9)  edge[->,bend right=-20, below] node {$a$} (6)

        (7)  edge[->,bend right=0, right] node {$b$} (81)
        ;
        \path[->,right] (7) edge[out=-145, in=-115, loop,below] node {$a$} ();
        \path[->,right] (6) edge[loop below,->,below] node {$a$} ();
        \path[->,right] (9) edge[loop right,->,above] node {$b$} ();
        \path[->,right] (8) edge[loop left,->,left] node {$b$} ();
        \path[->,right] (81) edge[loop left,->,left] node {$b$} ();

    \end{tikzpicture}
    \caption{The Associated Derivative Automaton of  \( \mathtt{ExtDist}(a^*b^* + b^*a^*, b^*a^*b^*, a^*b^*a^*)\). }%
    \label{ref: fig auto grad comb monad}
\end{figure}

However, notice that not every monadic expression produces a finite set of derivated terms, as shown in the next example.

\begin{example}
    Let us consider the expression \(E\) of Example~\ref{ex: calc der} and the expression \(F=E\cdot c^*\).
    It can be shown that
    \begin{align*}
        d_{a^n}(F) & = \mathtt{toExp}(d_{a^n}(E))\cdot c^*                                     \\
                   & = \mathtt{ExtDist}(a^*b^*+a^*, a^*b^*, a^*b^*a^*+ n \odot a^*) \cdot c^*.
    \end{align*}
\end{example}

The study of the necessary and sufficient conditions of monads that lead to a finite set of derivated terms is one of the next steps of our work.
\section{Haskell Implementation}\label{sec haskell}
The notions described in this paper have been implemented in Haskell, as follows:
\begin{itemize}
      \item The notion of monad over a sub-category of sets is a typeclass using the \emph{Constraint kind} to specify a sub-category;
      \item \(n\)-ary functions and their operadic structures are implemented using fixed length vectors, the size of which is determined at compilation using type level programming;
      \item The notion of graded module is implemented through an existential type to deal with unknown arities:
            Its monadic structure is based on an extension of heterogeneous lists, the \emph{graded vectors}, typed w.r.t.\ the list of the arities of the elements it contains;
      \item The parser and some type level functions are based on dependently typed programming with singletons~\cite{AW12}, allowing, for example, determining the type of the monads or the arity of the functions involved at run-time;
      \item An application is available \href{http://ludovicmignot.free.fr/programmes/monDer/index.html}{here}~\cite{AppWeb,gitLM} illustrating the computations:
            \begin{itemize}
                  \item the backend uses \href{https://hackage.haskell.org/package/servant}{servant} to define an API hosted \href{https://mon-test-stack.herokuapp.com}{by Heroku};
                  \item the frontend is defined using \href{https://hackage.haskell.org/package/reflex}{Reflex}, a functional reactive programming engine and cross compiled in JavaScript with GHCJS\@.
            \end{itemize}
            As an example, the monadic expression of the previous examples can be entered in the web application as the input \texttt{ExtDist(a*.b*+b*.a*,b*.a*.b*,a*.b*.a*)}.
\end{itemize}

\section{Conclusion and Perspectives}

In this paper, we achieved the first step of our plan to unify the derivative computation
over word expressions.
Monads are indeed useful tools to abstract the underlying computation structures
and thus may allow us to consider some other functionalities, such as
capture groups \emph{via} the well-known StateT monad transformer~\cite{MPJ95}, that we plan to study in a future work.
We also aim to study the conditions satisfying by monads that lead to finite set of derivated terms, and
to extend this method to tree expressions using enriched categories.

\bibliographystyle{eptcs}
\bibliography{biblio}

\end{document}